\documentclass[a4paper,UKenglish]{lipics-v2016}
 
\usepackage{microtype}
\usepackage[]{algorithm2e}
\usepackage{multirow}
\usepackage{float}

\title{Assembling Omnitigs using Hidden-Order de Bruijn Graphs\footnote{Partially supported by EU grant H2020-MSCA-RISE-2015 BIRDS690 No. 690941; by Basal Funds FB0001, Conicyt, Chile; by a Conicyt PhD Scholarship; and by Fondecyt Grant 1-171058, Chile.}}
\titlerunning{Assembling Omnitigs using Hidden-Order de Bruijn Graphs} 

\author[1,2]{Diego Díaz-Domínguez}
\author[3]{Djamal Belazzougui}
\author[1,4]{Travis Gagie}
\author[5]{Veli Makinen}
\author[1,2]{Gonzalo Navarro}
\author[5]{Simon J. Puglisi}

\affil[1]{CeBiB --- Center for Biotechnology and Bioengineering, Chile}
\affil[2]{Department of Computer Science, University of Chile, Chile.
  \texttt{\{diediaz,gnavarro\}@dcc.uchile.cl}}
\affil[3]{CERIST, Algeria.
  \texttt{djamal.belazzougui@gmail.com}}
\affil[4]{School of Computer Science and Telecommunications, Diego Portales University, Chile.
  \texttt{travis.gagie@gmail.com}}
\affil[5]{Department of Computer Science, University of Helsinki, Finland.
  \texttt{\{veli.makinen,simon.puglisi\}@helsinki.fi}}

\authorrunning{D. Díaz-Domínguez et al.} 

\Copyright{Diego Díaz-Domínguez, Travis Gagie, Djamal Belazzougui, Veli Makinen, Simon Puglisi and Gonzalo Navarro;}


\EventEditors{John Q. Open and Joan R. Acces}
\EventNoEds{2}
\EventLongTitle{42nd Conference on Very Important Topics (CVIT 2016)}
\EventShortTitle{CVIT 2016}
\EventAcronym{CVIT}
\EventYear{2016}
\EventDate{December 24--27, 2016}
\EventLocation{Little Whinging, United Kingdom}
\EventLogo{}
\SeriesVolume{42}
\ArticleNo{23}

\begin{document}

\maketitle

\begin{abstract}
{\it De novo} DNA assembly is a fundamental task in Bioinformatics, and finding Eulerian paths on de Bruijn graphs is one of the dominant approaches to it.  In most of the cases, there may be no one order for the de Bruijn graph that works well for assembling all of the reads. For this reason, some de Bruijn-based assemblers try assembling on several graphs of increasing order, in turn.  Boucher et al. (2015) went further and gave a representation making it possible to navigate in the graph and change order on the fly, up to a maximum $K$, but they can use up to $\lg K$ extra bits per edge because they use an LCP array.  In this paper, we replace the LCP array by a succinct representation of that array's Cartesian tree, which takes only 2 extra bits per edge and still lets us support interesting navigation operations efficiently.  These operations are not enough to let us easily extract unitigs and only unitigs from the graph but they do let us extract a set of safe strings that contains all unitigs.  Suppose we are navigating in a variable-order de Bruijn graph representation, following these rules: if there are no outgoing edges then we reduce the order, hoping one appears; if there is exactly one outgoing edge then we take it (increasing the current order, up to $K$); if there are two or more outgoing edges then we stop.  Then we traverse a (variable-order) path such that we cross edges only when we have no choice or, equivalently, we generate a string appending characters only when we have no choice.  It follows that the strings we extract are safe. Our experiments show we extract a set of strings more informative than the unitigs, while using a reasonable amount of memory.
 \end{abstract}

 \section{Introduction}
 \label{sec:introduction}

  \emph{De novo} DNA assembly is one of the oldest problems in Bioinformatics. It consists of inferring the sequences of a genome from a set of \emph{sequencing reads}. The classical approach is to first create a de Bruijn graph of order $K$ from the reads and then finding Eulerian paths in that graph. The choice of $K$ is always a major concern, especially when the sequencing coverage is not uniform. If $K$ is too short, then the graph becomes too dense and uninformative. On the other hand, if $K$ is too big, the graph becomes disconnected. There exist some tools \cite{Cha2016} that help choose a good $K$. However, in most cases, the Eulerian paths of maximum length cannot be obtained with just one $K$. For this reason, some de Bruijn-based assemblers, such as SPAdes~\cite{Bankevich2012} and IDBA~\cite{Peng2010}, try assembling on several graphs of increasing order, in turn.  Lin and Pevzner~\cite{Lin2014} went a step further and proposed manifold de Bruijn graphs, in which the vertices are strings of possibly many different lengths --- but they provide no practical implementation.  Boucher et al.~\cite{BBGPS15} offered an implementation and more by adding a longest common prefix (LCP) array to Bowe et al.'s~\cite{BOSS12} representation of de Bruijn graphs, based on the Burrows-Wheeler Transform (BWT), making it possible to navigate in the graph and change order on the fly (up to the order for which Bowe et al.'s underlying representation is built).  Whereas Bowe et al.'s representation takes only 4 bits per edge in the graph, however, Boucher et al.\ can use up to $\lg K$ extra bits per edge, where $K$ is its maximum order.

  In this paper, we replace the LCP array by a succinct representation of its Cartesian tree, which takes only 2 extra bits per edge and still lets us support interesting navigation operations efficiently: e.g., report the labels on the out-edge leaving a vertex, cross an out-edge with a given label (increasing the order by 1, up to $K$), and reduce the order until the set of distinct labels on the out-edges grows.  These operations alone do not let us easily extract unitigs and only unitigs from the graph (i.e., the labels of paths of vertices with in- and out-degree 1), but they do let us extract a set of safe strings that contains all unitigs.

  Tomescu and Medvedev~\cite{Tomescu2016} defined a string to be {\em safe} if it occurs in all valid reconstructions of the genome, and they defined {\em omnitigs} to be the safe strings that can be extracted from a particular de Bruijn graph.  Although they considered only standard de Bruijn graphs, with equal-length strings as the vertices, their definition generalizes easily to variable-order de Bruijn graph representations, such as Lin and Pevzner's or Boucher et al.'s.  Suppose we are navigating in a variable-order de Bruijn graph representation, following these rules: if there are no outgoing edges, then we reduce the order, hoping one appears; if there is exactly one outgoing edge, then we take it (possibly increasing the current order); if there are two or more edges, then we stop.  We traverse a (variable-order) path such that we cross edges only when we have no choice or, equivalently, we generate a string appending characters only when we have no choice.  It follows that the strings we extract are safe.

  Implementing our data structure is the first challenge, but there are others: when extracting unitigs, for example, there is no risk of being caught in cycles, or extracting the same substring over and over again as part of several unitigs, but we must deal with both these possibilities when following our rules, and varying the order makes the situation even more complicated.  Nevertheless, our experiments show we extract a set of strings more informative than the unitigs, while using a reasonable amount of memory.

 \subsection{Our contribution}

 In this paper, we address the problem of generating safe strings longer than unitigs. We demonstrate that those omnitigs spelt by walks in the assembly graph where the nodes have arbitrary indegree and outdegree one, which we call \emph{right-maximal} (RM) omnitigs,  can be efficiently obtained from a data structure that encodes a variable-order de Bruijn graph. For doing that, we first propose the \emph{hidden-order} BOSS data structure (HO-BOSS), a variation of VO-BOSS that uses less space by not storing the orders of the de Bruijn graph. Then, we propose a linear-time algorithm that works over HO-BOSS to retrieve the RM omnitigs. 
 \section{Preliminaries}

\noindent\textbf{De Bruijn graphs.}
A \emph{complete} de Bruijn Graph of order $K$, or $DBG_K$, is the labelled directed cyclic graph $G=(V, E)$ over the alphabet $\Sigma = [1..\sigma]$ consisting of $\sigma^{K}$ nodes and $\sigma^{K+1}$ edges. Each node $v \in V$ is labelled with a string $p \in \Sigma^{K}$ and every edge $e=(v, u) \in E$ represents the string $s \in \Sigma^{K+1}$ such that the label of $v$ is the prefix $s[1..K]$, the label of $u$ is the suffix $s[2..K+1]$, and the label of $e$ is the character $s[K+1]$.

The de Bruijn graph of order $K$ of a set of strings $\mathcal{S} = \{ S_1, S_2, \ldots, S_n\}$, or $DBG_K^{\mathcal{S}}$, is the subgraph of the complete $DBG_K$ induced by the edges that represent the set $\mathcal{K} \subseteq \Sigma^{K+1}$ of all $(K+1)$mers present in $\mathcal{S}$.

A \emph{variable order} $DBG_K^{\mathcal{S}}$, or $voDBG_K^{\mathcal{S}}$, is a graph
formed by the union of all the graphs $DBG_k^{\mathcal{S}}$, with $1\leq k \leq K$ for $\mathcal{S}$. Every $DBG_k^{\mathcal{S}}$ represents a \emph{context} of the $voDBG_K^{\mathcal{S}}$. In addition to the (directed) edges of each $DBG_k^{\mathcal{S}}$, two nodes $v$ and $v'$, with $v \in DBG_k^{\mathcal{S}}$, $v' \in DBG_{k'}^{\mathcal{S}}$, and $k > k'$, are connected by an undirected edge $(v, v')$ if the label of $v'$ is a suffix of the label of $v$. Following the edge $(v, v')$ from $v$ to $v'$, or vice-versa, is called a \emph{change of order}.

\bigskip\noindent\textbf{Rank and select data structures.}
\emph{Rank} and \emph{select} dictionaries are fundamental in most succinct data structures. Given a sequence $B[1..n]$ of elements over the alphabet $\Sigma=[1..\sigma]$, $B.\texttt{rank}_{b}(i)$ with $i \in [1..n]$ and $b\in\Sigma$, returns the number of times the element $b$ occurs in $B[1..i]$, while $B.\texttt{select}_b(i)$ returns the position of the $i$th occurrence of $b$ in $B$. For binary alphabets, $B$ can be represented in $n+o(n)$ bits so that \texttt{rank} and \texttt{select} are solved in constant time \cite{Cla96}. When $B$ has $m \ll n$ 1s, a compressed representation using $m\lg\frac{n}{m}+O(m)+o(n)$ bits, still solving the operations in constant time, is of interest \cite{RRR07}. This space is $o(n)$ if $m=o(n)$.

To support \texttt{rank} and \texttt{select} on non-binary alphabets, a popular data structure is the wavelet tree \cite{GGV03}. This structure can represent the sequence within its zero-order entropy (i.e., the entropy of the distribution of its symbols) plus a sublinear term, $o(n\log\sigma)$ bits, and offer access to any position in the sequence, as well as \texttt{rank} and \texttt{select} functionality, in time $O(\log\sigma)$. There are other representations offering compressed space and operations in time $O(\log\log\sigma)$ \cite{BCGNN14}, but they are not practical enough for the small alphabets of interest in sequence assembly.

\bigskip\noindent\textbf{Succinct representation of ordinal trees.}
An ordinal tree $T$ with $n$ nodes can be stored succinctly as a sequence of \emph{balanced parentheses} (BP) encoded as a bit vector $B=[1..2n]$. Every node $v$ in $T$ is represented by a pair of parentheses \texttt{(..)} that contain the encoding of the complete subtree rooted at $v$. Note that $B$ can be easily constructed with a \emph{depth first search} traversal over $T$. Every node of $T$ can be identified by the position in $B$ of its open parenthesis. 

A large number of navigational operations over $T$ can be simulated with a small set of primitives over $B$. These primitives are \texttt{enclose}, \texttt{open}, \texttt{close}, \texttt{rank} and \texttt{select}, where \texttt{enclose}$(i)$ gives the position of the rightmost open parenthesis that encloses $B[i]$ and \texttt{open$(i)$/close$(i)$} give the position of the parenthesis that pairs the close/open parenthesis at $B[i]$, respectively. 

Navarro and Sadakane \cite{NS14} showed that, for static trees, all these primitve operations, and consequently all the navigational operations built on them, can be answered in constant time using a data structure that needs $2n + o(n)$ bits of space.

\bigskip\noindent\textbf{BOSS representation for de Bruijn graphs.}
BOSS \cite{BOSS12} is one the most succinct data structures for encoding de Bruijn graphs. The process of building the BOSS index of the $K$-order de Bruijn graph of a set of strings $\mathcal{S}=\{S_1,S_2,\ldots\}$ is as follows: Pad every $S_i \in \mathcal{S}$ with $K$ copies of the symbol \$ and scan all the resulting strings, obtaining the set $\mathcal{K}$ of all their $(K+1)$mers. Then, sort $\mathcal{K}$ in lexicographic order from right to left, starting from their $K$th position, with ties broken by the symbols at position $K+1$. Notice that the first $K$ positions of all the $(K+1)$mers (from now on referred to as the BOSS matrix, with one row per $(K+1)$mer and one column per position in $1..K$) contain all the possible $K$mers present in $\mathcal{S}$, whereas the position $K+1$ (from now on referred to as the $E$ array) contains all the outgoing edges of $K$mers sorted in lexicographical order. To delimit the outgoing edges of every $K$mer, create a bit vector $B$ in which $B[i]=1$ if $E[i]$ is the last outgoing edge of some $K$mer in the BOSS matrix. To encode the incoming edges, store in every $E[i]$ the corresponding symbol $c$ of the edge if $E[i]$ is within a range of nodes suffixed by the same $K-1$ characters, and $E[i]$ is the first occurrence of $c$ in that range; otherwise store an alternative symbol $\overline{c}$ in $E[i]$. Finally, replace the BOSS matrix with an array $C[1..\sigma]$ with the cumulative counts of the symbols of the column $K$,
encode the $E$ array using a wavelet tree, and compress the $B$ vector with support for \texttt{rank} and \texttt{select} operations. 

Bowe et al.~\cite{BOSS12} proposed several navigational queries over the BOSS index, almost all of them with at most an $O(\log \sigma)$-factor slowdown. The most relevant for this paper are:

\begin{itemize}
    \item \texttt{outdegree$(v)$}: number of outgoing edges of $v$.
    \item \texttt{forward$(v, a)$}: node reached by following an edge from $v$ labelled with symbol $a$.
    \item \texttt{indegree$(v)$}: number of incoming edges of $v$.
    \item \texttt{backward$(v)$}: list of the nodes with an outgoing edge to $v$.
\end{itemize}

Boucher et al.~\cite{BBGPS15} noticed that by considering just the last $k$ columns in the BOSS matrix, with $k\leq K$, the resulting nodes are the same as those in the de Bruijn graph of order $k$. This property implies that, for the same dataset, all the de Bruijn graphs up to order $K$ can be implicitly encoded with the same BOSS index. To support the variable-order functionality, that is, changing the value of $k$ when necessary, they augmented BOSS with the \emph{longest common suffix} ($LCS$) array, and they call this new index the \emph{variable-order} BOSS (VO-BOSS). The $LCS$ array stores, for every node of maximal order $K$, the size of the longest suffix shared with its predecessor node in the right-to-left lexicographical ordering. Additionally, they defined the following operations:

\begin{itemize}
    \item \texttt{shorter$(v, k)$}: range of the nodes suffixed by the last $k$ characters of $v$. 
    \item \texttt{longer$(v, k)$}: list of the nodes whose labels have length $k \leq K$ and end with that of $v$. 
    \item \texttt{maxlen$(v, a)$}: a node in the index suffixed by the label of $v$, and that has an outgoing edge labelled with $a$.
\end{itemize}

Given a node $v$ of arbitrary order $k'\leq K$ represented as a range $[i, j]$ in the BOSS matrix, the operation \texttt{shorter$([i, j], k)$} can be implemented by searching for the largest $i'<i$ and the smallest $j' > j$ with $LCS[i'-1]$, $LCS[j']$ < $k$ and then returning $[i', j']$. For \texttt{larger$([i, j], k)$}, it is necessary to first obtain the set $B=\{b, i-1 \leq b \leq j \land LCS[b] <k \}$, and then for every consecutive pair $(b, b')$, report the range $[b + 1, b']$. The function \texttt{maxlen$([i, j], a)$} is the easiest: it suffices with searching in the range $[i, j]$ a node of maximum $K$ with an outgoing edge labelled with $a$.

\section{Variable order de Bruijn graphs and tries}

As mentioned in Section~\ref{sec:introduction}, some assemblers produce several de Bruijn graphs with different $K$s and then incrementally construct the contigs from those graphs.  An alternative option is to build the VO-BOSS index for the variable-order de Bruijn graph $voDBG_K^{\mathcal{R}}$. The advantage of VO-BOSS is that it allows changing the order of the nodes \emph{on the fly}, without generating different data structures for every $k \leq K$. Still, VO-BOSS increases the space complexity of the regular BOSS by an $O(\log K)$ factor, because it needs to encode the $LCS$ array of the $K$mers as a wavelet tree to support the change of order efficiently. This is an important issue, because if $K$ is too large, then memory requirements may increase significantly. Thus, a more compact way of representing $voDBG_K^{\mathcal{R}}$ is necessary.

During the traversal of $voDBG_K^{\mathcal{R}}$, if a node $v_i^k$ with no valid edges is reached, then the context can be changed from $k$ to $k'<k$ so the traversal can continue through the edges of the new node $v_i^{k'}$. However, not every change of order from $k$ to $k'$ generates new edges.

\begin{lemma}\label{l1}
Let $v_i$ and $v_j$ be two nodes in a variable-order de Bruijn graph, with labels $S_i$ and $S_j$, respectively, where $S_j$ is a proper suffix of $S_i$. Then, if no proper suffix of $S_i$ of the form $S\,S_j$ is left-maximal in $\mathcal{K}$, then the outgoing edges from $v_i$ and $v_j$ are the same. 
\end{lemma}

\begin{proof}
Let $S = aS'$, where $a \in \Sigma$. If $S'$ is not left-maximal in $\mathcal{K}$, then $S'$ is always preceded by $a$ in $\mathcal{K}$, and thus the set of occurrences of $S'$ and $S$ in $\mathcal{K}$ are equal. By applying this argument successively, starting from $S' = S_j$ and
ending with $S=S_i$, we have that the set of occurrences of $S_i$ and $S_j$ in $\mathcal{K}$ are the same. Therefore, the nodes $v_i$ and $v_j$ have the same outgoing edges.
\end{proof}

We use Lemma \ref{l1} for traversing $voDBG_K^{\mathcal{R}}$ efficiently. We synchronize the nodes in $voDBG_K^{\mathcal{R}}$ with the nodes in the \emph{compact trie} $T_K$ induced by the reversals of the $K$mers. In $T_K$, every leaf represents a $K$mer, and $label(t_i)$ denotes the left-maximal string spelled from the root to node $t_i$. Note that $label(t_i)$ is a suffix of the $K$mers represented by the leaves in the subtree rooted at $t_i$. Thus, to shorten the context of node $v_i^k$ in $voDBG_K^{\mathcal{R}}$, we can go to its corresponding node $t_i$ in $T_K$, move to the parent of $t_i$, say $t_j$, and finally go back to the corresponding node $v_i^{k'}$ in $voDBG_K^{\mathcal{R}}$. By Lemma \ref{l1}, the nodes between $t_i$ and $t_k$ that were compacted in $T_K$ cannot yield new edges.

Still, it could be that, even when $v_i^{k'}$ has more occurrences than $v_i^k$ in $\mathcal{K}$, it has no edges to new $K$mers. Thus, if $v_i^{k'}$ has no new edges with respect to $v_i^{k}$, we repeat the operation until finding a node $v_i^{k''}$ with a shorter context that does have new edges.

To lengthen the context, the process is similar, but instead of going to the parents in $T_K$, we have to traverse the subtree rooted at $t_i$.

\subsection{HO-BOSS index}

Changing the order of a node in $voDBG_K^{\mathcal{R}}$ from $k$ to $k'$ using $T_K$ does not require us to know the value of $k$ or $k'$. Thus, we can use just the \emph{topology} of $T_K$ for jumping from node $v$ to its parent node $v'$ (or vice-versa), where the string depths of $v$ and $v'$ have \emph{hidden} values $h$ and $h'$, respectively. We only know that $v_i^{h'}$ represents the longest suffix of $v_i^h$ with chances of continuing the graph traversal, and that is enough for assembly. For supporting this hidden-order feature, we augment BOSS with a bit vector $F$ that encodes the topology of $T_K$ in \emph{balanced parenthesis} format (BP). We call this new index \emph{hidden-order} BOSS (HO-BOSS).

We limit the size of the hidden order $h$ to the range $m \leq h \leq K$, where $m$ is the minimum order allowed by the user. This minimum order is aimed at preventing the context from becoming too short and thus uninformative for the assembly. Therefore, we modify $F$ by removing all those nodes with \emph{string} depth $< m$ and then connecting the resulting subtrees to a dummy root. In this way, whenever we reach the root in $F$, we know that the context became shorter than the threshold $m$. 

While VO-BOSS requires $O(n\log K)$ bits of space to encode the $LCS$ array, where $n=|\mathcal{K}|$, HO-BOSS stores only $F$, requiring at most $4n + o(n)$ bits (since $T_K$ has $n$ leaves and less than $n$ internal nodes). HO-BOSS also stores the $B$, $E$, and $C$ arrays of BOSS.

The bitvector $F$ can be built from the $LCS$ array, in worst-case time $O(nK)$ and $O(n + K\log n)$ bits of working space, by (easily) adapting an algorithm to build suffix tree topologies from the longest common prefix array ($LCP$, analogous to $LCS$) \cite[Sec.~11.5.4]{Nav16}.

\subsection{Operations over the HO-BOSS index}

Every node $v \in T_K$ can be identified in two ways: with the position $p_v$ of its open parenthesis in $F$, or with the range $[i,j]$ of leaves under the subtree rooted at $v$. The first method is referred to as the \emph{index} nomenclature and the second one as the \emph{range} nomenclature. As every leaf in $T_K$ is also a $K$mer, we can use the range nomenclature to link nodes in $T_K$ with ranges in the BOSS matrix.

We first define four primitives over $T_K$: \texttt{close}, \texttt{enclose}, \texttt{lca}, and \texttt{children}. All of them use the index nomenclature and will be used as base for more complex queries. They can all be computed in constant time \cite{NS14}.

\begin{itemize}
    \item \texttt{close$(v)$}: the position in $F$ of the closing parenthesis of $v$.
    \item \texttt{enclose$(v)$}: the parent of $v$.
    \item \texttt{lca$(v,v')$}: the \emph{lowest common ancestor} of $v$ and $v'$.
    \item \texttt{children$(v)$}: the number of children of $v$.
\end{itemize}

We also define two primitives for mapping nodes in $T_K$ to ranges in the BOSS matrix: 

\begin{itemize}
    \item \texttt{id2range$(v)$}: transforms a node in index nomenclature to its range form. 
    \item \texttt{range2id$(i, j)$}: transforms a node in range nomenclature to its index form.
\end{itemize}

To implement the two functions defined above efficiently, we have to include two extra functions, \texttt{rank\_leaf} and \texttt{select\_leaf}, which perform \texttt{rank} and \texttt{select} operations over the pattern ``\texttt{()}'' in $F$, which is implemented in constant time \cite{NS14}. Thus, both operations, $\texttt{id2range}(v) = [\texttt{rank\_leaf}(v),\texttt{rank\_leaf}(\texttt{close}(v))]$ and $\texttt{range2id}(i,j)=\texttt{lca}(\texttt{select\_leaf}(i),\texttt{select\_leaf}(j))$, require constant time too.

Finally, to traverse the variable-order de Bruijn graph, we redefine the following functions:

\begin{itemize}
    \item \texttt{shorter$(i, j)$}: the node $[i', j']$ that is the parent of node $[i, j]$.
    \item \texttt{forward$([i, j], a)$}: the node $[i' ,j']$ obtained by following the outgoing edge of $[i,j]$ labelled with $a$.
\end{itemize}

For \texttt{shorter$(i, j)$}, first we map the range $[i,j]$ to a node $v$ in $T_K$. We then obtain its parent node $v'$, and finally, map $v'$ back to a range $[i',j']$. The whole operation is expressed as $\texttt{id2range}(\texttt{enclose}(\texttt{range2id}(i,j)))$. Thus, operation \texttt{shorter} takes constant time.

\begin{algorithm}[t]
\small
\SetAlgoLined
\LinesNumbered
\KwData{$\textrm{HOBOSS}(B,E,C,F), [i, j], a$}
\KwResult{$\texttt{forward}([i,j],a)$ in range nomenclature}
$p \leftarrow B.\texttt{select}_1(i-1)+1$\;
$q \leftarrow B.\texttt{select}_1(j)$\;
$sp \leftarrow E.\texttt{rank}_a(p-1)+1$\;
$ep \leftarrow E.\texttt{rank}_a(q)$\;
$l \leftarrow ep - sp + 1$\;
$i' \leftarrow C[a] + sp$\;
$j' \leftarrow i' + l - 1$\;
$v_{i'} \leftarrow F.\texttt{select\_leaf}(i')$\;
$v_{j'} \leftarrow F.\texttt{select\_leaf}(j')$\;
\Return{{\rm $\texttt{id2range}(F.\texttt{lca}(v_{i'},v_{j'}))$}}; 
\BlankLine
\caption{Forward operation in HO-BOSS.}
\label{alg:forward}
\end{algorithm}

The operation \texttt{forward$([i, j], a)$} is detailed in Algorithm \ref{alg:forward}. In lines 1--2, the range $(p, q)$ in the $E$ array of the BOSS index is obtained. This range contains the outgoing edges for the $K$mers within $[i, j]$. Then, the outgoing range $[i', j']$ is inferred in lines 3--7 (see Bowe et al.~\cite{BOSS12}). Since $[i', j']$ might be a subrange of the complete outgoing node, we need to obtain the range of the lowest common ancestor for leaves $i'$ and $j'$, which is done in lines 8--10. Notice that, if $[i, j]$ represents string $P$, then $\texttt{forward}([i, j], a)$ will return the range for string $Pa$. This is different from the classical \texttt{forward} operation in BOSS \cite{BOSS12}, in which the label of an outgoing node is of the same size as the current node.

Operation \texttt{forward} takes $O(\log \sigma)$ time, as it involves a constant number of operations over a wavelet tree, bitvectors, and parentheses. 

With these operations, we have a powerful machinery to build omnitigs in $voDBG_K^{\mathcal{R}}$.

\section{Assembling right-maximal omnitigs}

\emph{Right-maximal} (RM) omnitigs are those safe strings spelt by walks in $voDBG_K^{\mathcal{R}}$ with arbitrary indegree and outdegree strictly 1. RM omnitigs are expected to be present in circular and non-circular genomes and to be longer than unitigs. They are, however, more complicated to construct, because some RM omnitigs can be left-extensions of other RM omnitigs, so if we are not careful, we may report redundant sequences.

The first challenge is how to report just the \emph{longest} RM omnitigs, that is, those that are not suffixes of other RM omnitigs. This problem can be solved by just considering those walks in $voDBG_K^{\mathcal{R}}$ whose first node is a leaf in $F$ (i.e., a complete $K$mer) with outdegree 1 and with all its incoming nodes with outdegree at least 2. We call these leaves the \emph{starting nodes}. Notice that the starting nodes can be obtained by scanning the $K$mers and performing \texttt{outdegree} and \texttt{backward} queries on the regular BOSS index. 

The second challenge is how to avoid retraversing paths that were visited before. For instance, consider two RM omnitigs, $aP$ and $bP$. The walk spelling $P$ is traversed twice. To solve this, we mark all the \emph{path-merging} nodes in $voDBG_K^{\mathcal{R}}$, that is, nodes with \texttt{outdegree} 1 and \texttt{indegree} at least 2. For the marking process, we perform a depth-first traversal over $F$, and fill a bit vector $PM$ of size at most $2n$, so that $PM[i]=1$ if the $i$th node in $F$ is a path-merging node. We build the \texttt{rank} structures on $PM$. Associated with the 1s of $PM$, we will store in a vector $V$ a pointer to the position in the previous omnitig that was being generated when we passed through node $i$, so that we make our current omnitig point to it instead of generating it again. Initially, $V$ has all $null$ values.

\begin{algorithm}[t]
\small
\SetAlgoLined
\LinesNumbered
\KwData{$nodeRange$, HOBOSS, $PM$, $V$}
\KwResult{Omnitig $O$}
$O \leftarrow$ empty linked list of nodes\;
$nodeRank \leftarrow F.\texttt{rank}_{'('}(\texttt{range2id}(nodeRange))$\;
\While{true}{
    \If{$PM[nodeRank]=1$}{
       $p \leftarrow PM.\texttt{rank}_1(nodeRank)$\; \eIf{$V[p]\not=null$}{
            link $O$ with previous omnitig position, $V[p]$\; 
            {\bf break}\; 
        }{
            $V[p] \leftarrow$ current (last) list node of $O$\;
        }
    }
    \While{\rm $nodeRange$ has only $\$$s}{
        $nodeRange \leftarrow \texttt{shorter}(nodeRange)$\;
    }
    \eIf{\rm $nodeRange$ has only $a$s, $\overline{a}$s, and $\$$s, for some $a$}{
        append $a$ as a list node at the end of $O$\;
        $nodeRange \leftarrow \texttt{forward}(nodeRange,a)$\;
    }{
        {\bf break}\;
    }
}
\Return{$O$}
\BlankLine
\caption{Finding the longest RM omnitigs starting at a node.}
\label{alg:rmomnitigs}
\end{algorithm}

Algorithm \ref{alg:rmomnitigs} describes the extension of an RM omnitig $O$, which starts empty. In line 2, the rank of the input starter node is obtained. In lines 3--22 the algorithm continues the walk as long as it is not retraversing a path and the nodes have outdegree 1. In lines 4--12 we handle the visited nodes. First, we check with $PM$ if the current node is path-merging. If so, we check in $V$ if it was visited before, in which case we just link $O$ to the corresponding position of the previous omnitig and finish. If, instead, it is the first time we visit this node, we mark it as visited by associating it with the current node of $O$, and continue the traversal. In lines 13--21, we do the forward process. First, in lines 13--15, we shorten the node with \texttt{shorter}, as much as necessary to ensure there is some outgoing edge, that is, there is some non-\$ in $E[nodeRange]$. Once $nodeRange$ contains some edge, we see in lines 16--21 if it contains just one symbol, say $a$. If so, we append that $a$ to $O$ and move on to the next node using \texttt{forward}. If, instead, $nodeRange$ has more than one outgoing symbol, the generation of $O$ is complete. The result of generating all the omnitigs for all the starter nodes is a set of lists, where some can possibly merge into others. This saves both space and generation time.

The missing components are how to determine that $nodeRange$ contains only \$s, or only $a$s apart from $\$$s. To do this efficiently, we represent $E$ in a different way: a bit vector $BE$ so that $BE[i]=1$ iff $E[i]\not=\$$, and a reduced string $E'$ so that, if $E[i]\not=\$$, then $E'[BE.\texttt{rank}_1(i)]=E[i]$. All the operations on $E$ are easily translated into $BE$ and $E'$: apart from accessing it, which we have just described, we have $E.\texttt{rank}_a(i)=E'.\texttt{rank}_a(BE.\texttt{rank}_1(i))$ if $a\not=\$$ and
$E.\texttt{rank}_\$(i)=BE.\texttt{rank}_0(i)$, and $E.\texttt{select}_a(j)=BE.\texttt{select}_1(E'.\texttt{select}_a(j))$ if $a \not= \$$
and $E.\texttt{select}_\$(j)=BE.\texttt{select}_0(j)$. With this representation we have that there are only \$s in $nodeRange=[i,j]$ iff $BE.\texttt{rank}_1(j)= BE.\texttt{rank}_1(i-1)$. If there are some non-\$s, we map the range $E[i..j]$ to $E'[i',j']$, where $i'=BE.\texttt{rank}_1(i-1)+1$ and $j'=BE.\texttt{rank}_1(j)$. Then $a$ is an outgoing symbol, where $E'[i']=a$ or $E'[i']=\overline{a}$. It is the only outgoing symbol iff $(\texttt{rank}_a(E',j')-\texttt{rank}_a(E',i')) +(\texttt{rank}_{\overline{a}}(E',j')-\texttt{rank}_{\overline{a}}(E',i')) =j'-i'$. Thus, every operation in Algorithm \ref{alg:rmomnitigs} takes time $O(\log\sigma)$.

\section{Experiments}

\noindent\textbf{Datasets, code, and machine.}
Three Illumina datasets of single-end reads of 150 characters long were simulated from the reference genome of \emph{Escherichia coli} str. K-12 substr. MG1655, with coverages 5x, 10x, and 15x. The program used to simulate the reads was \texttt{wgsim} \cite{Li2012}, using error rate 0, mutation rate 0, and discarding regions which contain characters \texttt{N}. The sizes for the datasets are 62 MB for 5x, 124 MB for 10x, and 185 for 15x. 

 Both the HO-BOSS data structure and the algorithm for building RM omnitigs were implemented \footnote{https://bitbucket.org/DiegoDiazDominguez/boss-assembly} in \texttt{C++}, on top of the \texttt{SDSL-lite} \footnote{https://github.com/simongog/sdsl-lite} library. 

 All the experiments were carried out on a machine with Debian 4.9, 252 GB of RAM and procesor Intel(R) Xeon(R) Silver @ 2.10GHz, with 32 cores.

\bigskip\noindent\textbf{Memory and time usage.}
We generated 16 HO-BOSS indexes combining different values for $K$ (15, 20, $\ldots$, 50) and for $m$ (5, 10). We measured for every index its size (Figure \ref{fig:omnires}A), the memory peak (Figure \ref{fig:omnires}B), and the elapsed time (Figure \ref{fig:omnires}C) during the RM omnitigs construction. Additionally, we evaluated the fraction that each index represents with respect to the size of the original dataset (Figure \ref{fig:omnires}D). The size of the index and memory peak grow linearly with $K$. The former is about $0.8K\%$ and $1.0K\%$ of the original data size, whereas the latter is closer to $10K\%$. Indeed, our index amounts to only 7\%--14\% of the total space required by the process to find the longest omnitigs (see table \ref{tab:mem}). The table shows that, if we used $V$ as a bit vector (the minimum required to detect cycles in the omnitigs) and charge all the pointers and strings to the omnitig structure, then it takes 84\%-92\% of the space.

\begin{figure}[t]
  \includegraphics[width=\linewidth]{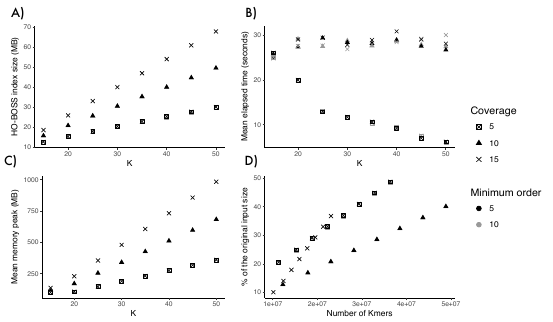}
  \caption{Statistics about the HO-BOSS index and the RM omnitigs algorithm. Each color represents a dataset and each shape a value of $m$ (both shapes are often superimposed). A) size of the HO-BOSS index vs $K$. B) mean elapsed time for the construction of RM omnitigs vs $K$. C) mean memory peak during the RM omnitigs construction vs $K$. D) Percentage of the original input data set (read sets) that each HO-BOSS index represents, as a function of the number of $K$mers.}
  \label{fig:omnires}
\end{figure}

On the other hand, the elapsed time does not change with $K$ for the coverages 10x and 15x, but it decreases with $K$ for 5x. The latter is due to the topology of the de Bruijn graph: as $K$ increases, the number of edges decreases, together with the number of possible walks spelling RM omnitigs. This is not the case for coverages 10x and 15x, because greater coverages prevent the reduction of edges. 

Figure \ref{fig:dist} shows the fraction of space that every substructure of the index uses. As expected, the $E$ array (the edges of the de Bruijn graph) and the bit vector encoding $T_K$ are the structures that use the most memory. The space distribution is stable across experiments. 

Table \ref{tab:omnivsuni} shows statistics about the structure of the HO-BOSS index. The number of $K$mers (i.e., graph nodes) increases more or less linearly with $K$, and also with the coverage. On the other hand, the number of starter nodes decreases with $K$, very sharply when moving from $K=15$ to $K=20$. This is expected since nodes with outdegree 1 are more common with larger $K$, and their targets are invalidated as starter nodes. The number of path-merging (PM) nodes also increases with $K$. The fraction of PM nodes with respect to the number of $K$mers also increases with $K$, and with the coverage. The percentages range from 15\% with $K=15$ and 5x, to 80\% with $K=50$ and 15x. This fraction impacts on the size of the HO-BOSS index, because the size of the array $V$ depends on the number of PM nodes. 

\begin{table}[t]
\centering
\begin{tabular}{cccccccc}
\hline
Dataset & $K$ & \# $K$mers & \# starter $K$mers & \# PM nodes & Max RM omnitig & Max unitig \\ 
\hline
\multirow{4}{*}{5x} &  15 & 10120448 & 505884 & 1568804 & 295 & 118 \\ 
&  20 & 12371884 & 6956 & 3387628 & 21344 & 134 \\ 
&  30 & 15959968 & 2674 & 6962222 & 31800 & 134\\ 
&  50 & 23070324 & 1156 & 14060718 & 49135 & 134\\ 
\hline
\multirow{4}{*}{10x} &  15 & 11286124 & 489177 & 2333984 & 295 & 107 \\ 
&  20 & 15262160 & 6676 & 5860764 & 35421 & 124 \\ 
&  30 & 22308302 & 2553 & 12890652 & 127947 & 124 \\ 
&  50 & 36379537 & 1111 & 26949353 & 269665 & 124 \\ 
\hline
\multirow{4}{*}{15x} &  15 & 12246035 & 469933 & 3063092 & 295 &  73 \\ 
&  20 & 17859860 & 6342 & 8216992 & 35420 &  86 \\ 
&  30 & 28196394 & 2432 & 18534531 & 127947 &  86 \\ 
&  50 & 48858570 & 1053 & 39181770 & 269665 &  86 \\ 
\hline
\end{tabular}
\caption{Statistics about the HO-BOSS indexes, for $m=10$. Column Max RM omnitig shows the size of the maximum RM omnitig obtained for that index, while column Max unitig is the size of the maximum unitig obtained with the BOSS indexes with $K$mer size ranging between $m$ and $K$. Some rows were omitted to save space.}
\label{tab:omnivsuni}
\end{table}

Table \ref{tab:performance} measures the time and space usage during the construction of the omnitigs. We show the effective number of traversed nodes (those reached by performing the forward operation in the graph) and the total number of nodes in the omnitigs (number of nodes reached by forward operation plus number of nodes linked in omnitigs due to path merging). The number of forward operations saved via path-merging is greater for low coverage (5x) and greater in general as $K$ increases. The average time and working space spent per node remains similar in all the experiments. We note that the savings are not too high, although some method is needed anyway to detect repeated traversals because otherwise a traversal may fall in a loop of the graph.


\bigskip\noindent\textbf{Quality of assembled omnitigs.}
We assessed the quality of the assembled RM omnitigs by mapping them back to the original reference genome using BLAST \cite{AGMML90}. An RM omnitig was classified as \emph{correctly assembled} if 100\% of its sequence could be aligned to the reference genome, with neither mistmatches nor gaps. Figure \ref{fig:syn_frac} shows the fraction of RM omnitigs correctly assembled in each experiment. Additionally, Figure \ref{fig:good_rm_sizes} shows the length distribution of the RM omnitigs correctly assembled. As expected, the results improve with higher coverage, being almost perfect for 10x and 15x. On 5x, we note that $K=15$ yields the best results, although the omnitigs generated are much shorter (see Table~\ref{tab:omnivsuni}).

\begin{figure}[t!]
\centering
  \includegraphics[width=0.6\linewidth]{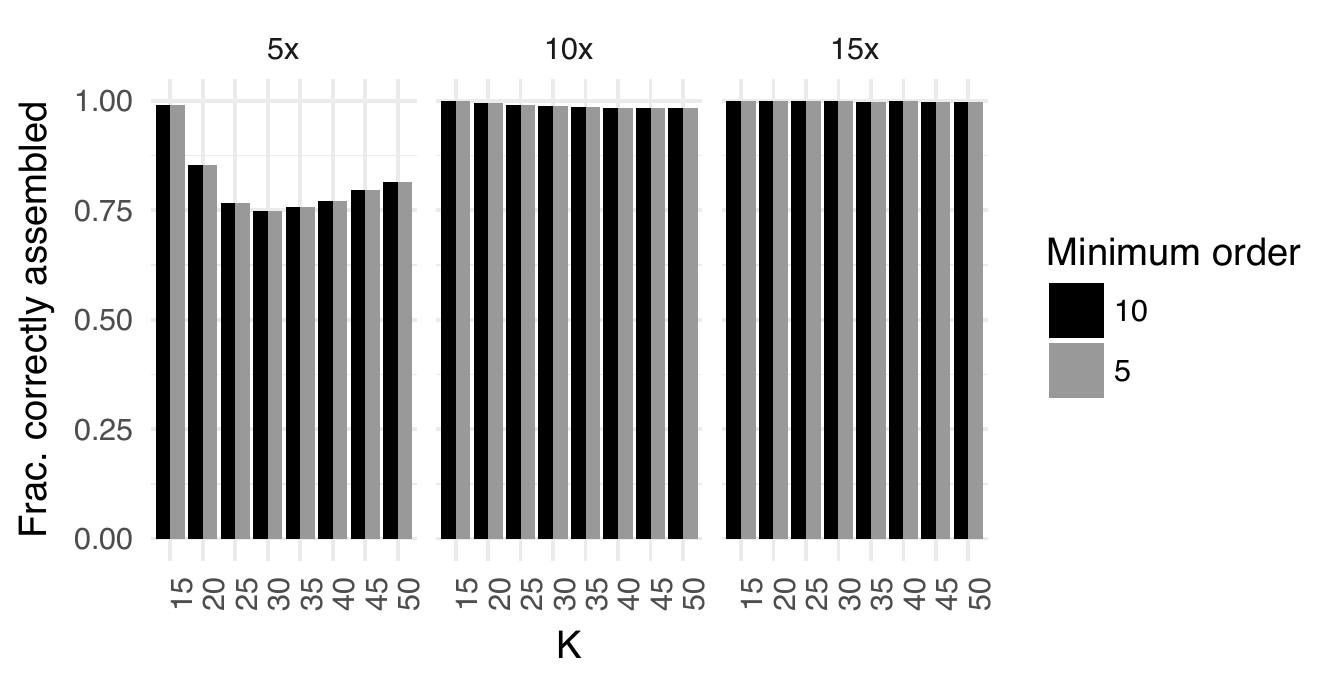}
  \caption{Omnitigs correctly assembled for coverages 5x, 10x, and 15x. The x-axis is the value for $K$ used to build the HO-BOSS index and the y-axis the fraction of omnitigs that were correctly assembled. Colors refer to the minimum value $m$. }
\label{fig:syn_frac}
\end{figure}

\begin{figure}[t!]
\centering
  \includegraphics[width=0.7\linewidth]{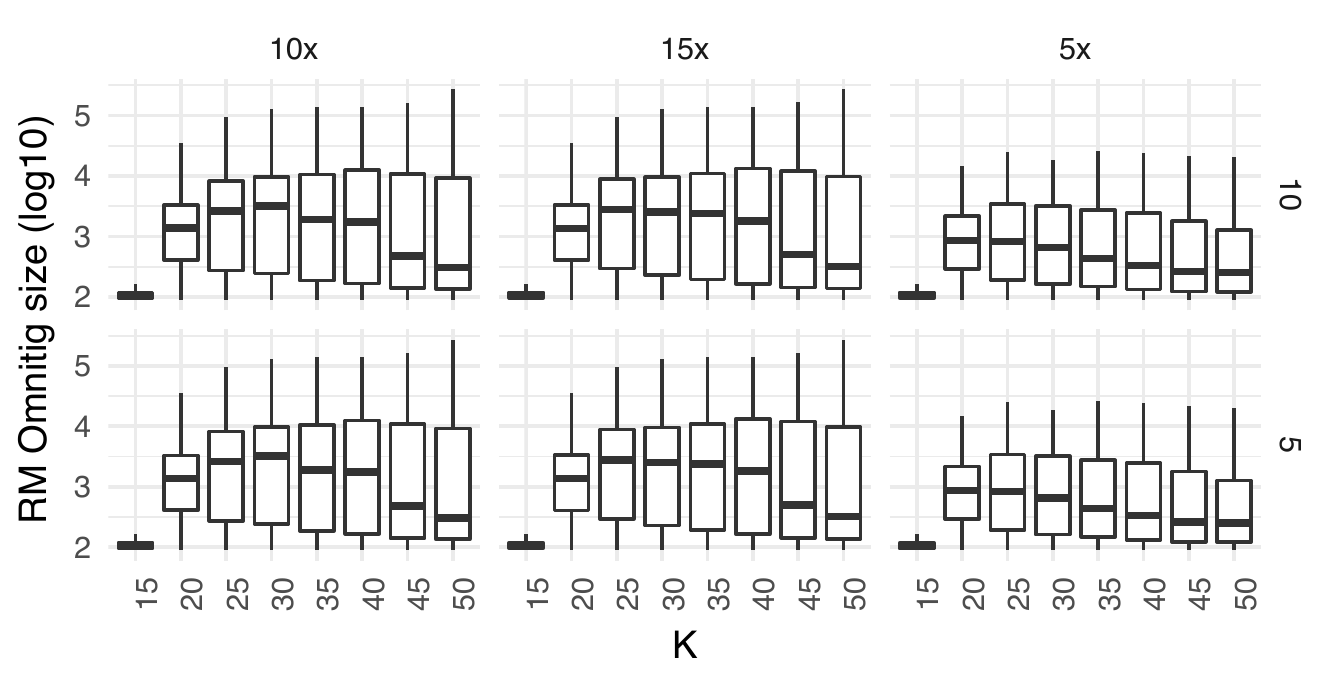}
  \caption{Length distribution of correctly assembled omnitigs. The x-axis gives the values of $K$ used to build the HO-BOSS indexes and the y-axis shows the length of the omnitigs. Plots in the columns represent the different datasets and plots in the rows represent the different values for $m$.}
\label{fig:good_rm_sizes}
\end{figure}

\bigskip\noindent\textbf{Comparing omnitig vs unitig lengths.}
We compare the sizes of the correctly assembled omnitigs and their equivalent unitigs using the following procedure. First, we built 40 regular BOSS indexes for each dataset, with $K$ ranging from 10 to 50, and the unitigs were inferred independently in those indexes. Subsequently, for every HO-BOSS index with $m=10$ and $15\leq K\leq 50$, we compare its RM omnitig of maximum size with all the unitigs of maximum size obtained in the regular BOSS indexes with $10 \leq K' \leq K$.

The sizes of maximum omnitigs versus maximum unitigs are shown in Table \ref{tab:omnivsuni}. In Figure \ref{fig:good_rm_sizes} we can see that the longest unitigs are near $\log_{10}length=2$, below most of the omnitigs we generate and about an order of magnitude below the mean omnitig length.  


\section{Conclusion and further work}

We have presented HO-BOSS, an approach to assembling a subclass of omnitigs on variable-order de Bruijn graphs without using and LCP array, and our experimental results show that we can extract remarkably long omnitigs that map perfectly to the source genomes.  Our approach still needs work to handle, for example, sequencing errors or paired-reads.

We hope HO-BOSS will serve as a starting point for the development of succinct genome assembers that produce longer sequences than current methods.  In order for this hope to be realized, we are working on more efficient algorithms to build the BOSS index, generate a {\em canonical} representation of the $K$mers (linking the information of a $K$mer with it reverse complement), and resolve problematic substructures in the graph such as bubbles and tips.

Our attempt to avoid retraversing suffixes of omnitigs already generated seems to reduce the amount of work only by a small margin (less than 15\%), but it makes the process require an order of magnitude more memory than the succinct HO-BOSS index. In the final version of this paper we will investigate an alternative that uses much less space: $V$ is converted into a bit vector that marks the nodes visited {\em for the current omnitig}, with the purpose of avoiding cycles only; no attempt is made to detect repeated suffixes of other omnitigs. 

\bibliography{hoboss}
\appendix
\setcounter{table}{0}
\renewcommand{\thetable}{A\arabic{table}}

\newpage

\section{Appendix}

\begin{table}[h]
\centering
\begin{tabular}{ccccc}
\hline
Dataset& K   & \%HO-BOSS   & \%V & \%Omni\\
\hline
\multirow{4}{*}{5x} & 15  & 13.172 & 1.104 & 85.724\\
& 20  & 12.135 & 1.491 & 86.374\\
& 30  & 8.942 & 1.349 &  89.709\\
& 50 & 7.254 & 1.228 & 91.517\\
\hline
\multirow{4}{*}{10x} & 15 & 13.497 & 1.153 & 85.350\\
& 20 & 11.236 & 1.417 & 87.347\\
& 30 & 8.332 & 1.280 & 90.388\\
& 50 & 6.899 & 1.180 & 91.921\\
\hline
\multirow{4}{*}{15x} & 15 & 12.558 & 1.004 & 86.438\\
& 20 & 14.360 & 1.626 & 84.014\\
& 30 & 10.722 & 1.525 & 87.753\\
& 50 & 8.334 & 1.365 & 90.301\\
\hline
\end{tabular}
\caption{Distribution of memory across the data structures during the memory peak of the omnitigs assembly. We show the percentage of memory occupied by the HO-BOSS index, by $V$ and by the pointers and strings. Here, $V$ is taken as a bit vector and its pointer is charged to \%Omni.}
\label{tab:mem}
\end{table}

\begin{table}[h]
\centering
\begin{tabular}{ccccccc}
\hline
Dataset & $K$ & traversed nodes & omnitig nodes & \% reduction & $\mu$seconds per node \\ 
\hline
\multirow{4}{*}{5x} & 15 & 7922657 & 8019021& 1.20 & 2.693\\
& 20 & 6275049 & 7097839& 11.59 & 2.684  \\
& 30 & 3670985 & 4298129& 14.59 & 2.652  \\
& 50 & 2139434 & 2509511& 14.75 & 2.740  \\
\hline
\multirow{4}{*}{10x} & 15 & 7659722 & 7743926& 1.09 & 2.959 \\
& 20 & 8308020 & 8399299& 1.09 & 2.949 \\
& 30 & 8264133 & 8489031& 2.65 & 3.092 \\
& 50 & 7592128 & 8022769& 5.39 & 3.083  \\
\hline
\multirow{4}{*}{15x} & 15 & 7370858 & 7448419& 1.04 & 2.880 \\
& 20 & 8079033 & 8120524& 0.51 & 3.029 \\
& 30 & 7637694 & 7688838& 0.67 & 3.122  \\
& 50 & 7703818 & 7904738& 2.54 & 3.210  \\
\hline
\end{tabular}
\caption{Performance statistics. The third column shows the effective number of traversed nodes in the variable order de Bruijn graph. The fourth column shows the total number of nodes that the assembled RM omnitigs represent. The fifth column shows the percentage of the total number of nodes that were not traversed due to our merging approach and the sixth and seventh columns show the average time and working space used per traversed node. Some rows were omitted to save space.}
\label{tab:performance}
\end{table}

\begin{figure}[h]
\centering
  \includegraphics[width=0.95\linewidth]{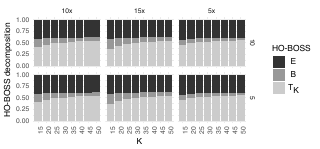}
  \caption{Fraction of memory that every substructure of the HO-BOSS index occupies. Plots in the columns represent the datasets while the plots in the rows represent the different values for $m$. The y-axis is the fraction of memory used by every data structure and the x-axis is the value used for $K$.}
  \label{fig:dist}
\end{figure}

\end{document}